\documentclass{llncs}

\newcommand\new{\color{red}}

\newcommand\remove[1]{}

\usepackage{amsmath, amssymb, amsfonts}
\usepackage{paralist}
\usepackage{alltt}
\usepackage{wrapfig}
\usepackage{stmaryrd}

\usepackage{caption}
\usepackage{subcaption}
\usepackage{floatrow}
\usepackage{calc}% To calculate width for \FBwidth
\usepackage{algorithm}
\usepackage[noend]{algorithmic}

\usepackage{tikz}
\usetikzlibrary{matrix}
\usetikzlibrary{fit}
\usetikzlibrary{automata}
\usetikzlibrary{shapes}

\usepackage{pbox}
\usepackage{mathtools}
\usepackage{footmisc}
\usepackage{xspace}
\usepackage{color}

\newcommand\comment[1]{}

\newcommand\Vars{\ensuremath{\mathcal{V}}}

\newcommand\Prog{\ensuremath{\mathit{P}}}
\newcommand\PProg{\ensuremath{\mathcal{P}}}
\newcommand\SemPreGen[1]{\ensuremath{\mathit{SemPreGen}(#1)}}

\newcommand\ProgCons{\ensuremath{\Phi}}
\newcommand\ProgState{\ensuremath{S}}

\newcommand\Valuation{\ensuremath{\mathcal{D}}}
\newcommand\errVar{\ensuremath{\mathit{err}}}

\newcommand\trace{\pi}

\newcommand\LearnGood{\mathit{LearnGood}}
\newcommand\LearnGoodUnder{\mathit{LearnGoodUnder}}
\newcommand\FixBad{\mathit{FixBad}}

\newcommand{\redn}{\mathit{DFAsserts}}
\newcommand{\blackn}{\mathit{IntraThreadOrder}}
\newcommand{\bluen}{\mathit{DFConds}}
\newcommand{\pinkn}{\mathit{NonFreeOrder}}
\newcommand{\readn}{\mathit{Read}}
\newcommand{\writen}{\mathit{write}}

\newcommand{\last}{\ensuremath{\mathit{last}}\xspace}
\newcommand{\depends}{\ensuremath{\mathit{depends}}\xspace}

\newcommand{\interfere}{\mathit{interfere}}

\newcommand\switch[2]{#1 \ensuremath{\leftrightsquigarrow} #2}
\newcommand\reorder{\theta}

\newcommand\arsays[1]{{\bf AR: #1}}

\newcommand\mypara[1]{\noindent{\bf #1}}

\begin{document}
\title{Regression-free Synthesis for Concurrency\thanks{This 
    research was funded in part by the European Research Council (ERC)
    under grant agreement 267989 (QUAREM), by the Austrian Science
    Fund (FWF) project S11402-N23 (RiSE), and by a gift from Intel
    Corporation. NICTA is funded by the Australian Government through
    the Department of Communications and the Australian Research Council
through the ICT Centre of Excellence Program.}} 
\author{Pavol {\v C}ern{\'y}\inst{1} \and Thomas A. Henzinger\inst{2}
\and Arjun Radhakrishna\inst{2} \and Leonid
Ryzhyk\inst{3}\inst{4}\and Thorsten Tarrach\inst{2}}
\institute{University of Colorado Boulder \and IST Austria 
\and University of Toronto \and NICTA, Sydney, Australia$^*$}
\maketitle

% \arsays{TODO:
  % \begin{itemize}
    % \item~Fix all the references to red/blue/black/pink edges. The paper
      % should be readable in colourless printer
    % \item~Remove all references to the complexity lemma
    % \item~Add the trace transformations into blue edges and proofs
    % \item~change x=1 to either x:=1 or x==1
  % \end{itemize}
% }
\begin{abstract}
While fixing concurrency bugs, program repair algorithms may introduce new
concurrency bugs.
We present an algorithm that avoids such regressions.
The solution space is given by a set of program transformations we consider
in the repair process. 
These include   
reordering of instructions within a thread and inserting atomic sections. 
%
%The new algorithm, PACES, is an extension of the CEGIS loop. 
%
The new algorithm learns a constraint on the space of candidate solutions, 
from both positive examples (error-free traces) and counterexamples (error
traces).    
From each counterexample, the algorithm learns a constraint necessary to
remove the errors.
From each positive examples, it learns a constraint that is
necessary in order to prevent the repair from turning the trace into an
error trace. 
We implemented the algorithm and evaluated it on simplified Linux
device drivers with known bugs.  
%
%The tool is able to fix the bugs while avoiding regressions.
\end{abstract}

%\ttsays{We should adjust how we talk about PACES in the abstract}

\section{Introduction}
\label{sec:intro}

The goal of program synthesis is to simplify the programming task by 
letting the programmer specify (parts of) her intent declaratively.
{\em Program repair} is the instance of synthesis where we are given both a 
program and a specification. 
The specification classifies the execution of the program into {\em good
traces} and {\em bad traces}. 
The synthesis task is to automatically modify the program so that the 
bad traces are removed, while (many of) the good traces are preserved. 

In {\em program repair for concurrency}, we assume that all errors 
are caused by concurrent execution.
We formalize this assumption into a requirement that all preemption-free
traces are good. 
The program may contain concurrency errors that are triggered by more
aggressive, preemptive scheduling. 
Such errors are notoriously difficult to detect 
% without a formal analysis 
and, in extreme cases, may only show up after years of operation 
of the system. 
Program repair for concurrency allows the programmer to focus on the
preemption-free correctness, while putting 
the intricate task of proofing the code for concurrency to the
synthesis tool.

\noindent
{\bf Program repair for concurrency.} 
The specification is provided by assertions placed by the programmer 
in the code. 
%Other properties like deadlock-freedom can be modelled using assertions.
A trace, which runs without any assertion failure, is called ``good'',
and conversely a trace with an assertion failure is ``bad''.
We assume that the good traces specify the intent of the programmer.
% Conversely, a trace is ``bad'' if it 
% violates the specification. 
A trace is complete if every thread finishes its execution. 
A trace of a multi-threaded program is preemption-free if a thread is 
de-scheduled only at preemption-points, i.e., when a thread tries to
execute a blocking operation, such as obtaining a lock. 

Given a multithreaded program in which all complete preemption-free
traces are good, the  program repair for concurrency problem is to find
a program for which the following two conditions hold: 
(a)~all bad traces of the original program are removed; and 
(b)~all the complete preemption-free traces are preserved.
We further extend this problem statement by saying that if not all
preemption-free traces are good, but all complete sequential traces are
good, then we need to find a program such that (a)~holds, and all
complete sequential traces are preserved.

\noindent
{\bf Regression-free algorithms.} Let us consider a trace-based algorithm
for program repair, that is, an iterative algorithm that in each
iteration is given a trace (good or bad) of the program-under-repair,
and produces a new program based on the traces seen. 
We say that such an algorithm is {\em regression-free} if after every
iteration, we have that: first, all bad traces examined so far are
removed, and second, all good traces examined so far are not turned
into bad traces of the new program. (Of course, to make this
definition precise, we will need to define a correspondence between
traces of the original program and the new program.)

\noindent
{\bf Program transformations.}
In order to remove bad traces, we apply the following program
transformations: (1)~reordering of adjacent instructions $i_1;i_2$ within 
a thread if the instructions are sequentially independent (i.e., if 
$i_1;i_2$ is sequentially equivalent to $i_2;i_1$), and (2)~inserting
atomic sections. 
The reordering of instructions is given priority as it may result in a 
better performance than the insertion of atomic sections. Furthermore, 
the reordering of instructions removes a surprisingly large number of
concurrency bugs that occur in practice; according to a study of how
programmers fix concurrency bugs in Linux device drivers~\cite{cav2013},
reordering of instructions is the most commonly used.  
%If no reordering helps, we turn to atomic sections to prevent 
%atomicity violations (see~\cite{lucia-pldi} for a classification of
%concurrency bugs).

\noindent
{\bf Our algorithm.}
Our algorithm learns constraints on the space of candidate solutions
from both good traces and bad traces.
We explain the constraint learning using as an example the program 
transformation~(1), which reorders instructions within threads. 
From a bad trace, we learn reordering constraints that eliminate 
the counterexample using the algorithm of~\cite{cav2013}.  
While eliminating the counterexample, such reorderings may transform a
(not necessarily preemption-free) good trace 
% without assertion failures 
into a bad trace
% along which some assertion fails 
--- this would constitute a regression.
In order to avoid regressions, our algorithm learns also from 
good traces. 
Intuitively, from a good trace $\trace$, we want to learn all the ways 
in which $\trace$ can be transformed by reordering without turning it into
an error trace--- this is expressed as a program constraint.
The program constraint is
\begin{inparaenum}[(a)]
\item~sound, if all programs satisfying the constraint are
  regression-free; and
\item~complete, if all programs violating the constraint have
  regressions.
\end{inparaenum}
However, as learning a sound and complete constraint is
computationally expensive, given a good trace $\trace$ we learn a sound
constraint that only guarantees that $\trace$ is not transformed into a
bad trace. We generate the constraint using data-flow analysis
on the instructions in $\trace$. The main idea of the 
analysis is that in good traces, the data-flow into passing assertions
is protected by synchronization mechanisms (such as locks) and data-flow
into conditionals along the trace. 
This protection may fail if we reorder instructions. We thus find
a constraint that prevents such bad reorderings. 
 
Summarizing, as the algorithm progresses and sees a set of bad
traces and a set of good traces, it learns constraints that encode the ways
in which the program can be transformed in order to eliminate the 
bad traces without turning the good traces into bad traces of the
resulting program. 

\noindent
{\bf CEGIS vs PACES.}
%\paragraph{CEGIS vs PACES.} 
A popular recent approach to synthesis is counterexample-guided 
inductive synthesis (CEGIS)~\cite{asplos06}.
Our algorithm can be viewed as an instance of CEGIS with the important
feature that we learn from positive examples.
We dub this approach PACES, for {\em Positive- and Counter-Examples in
Synthesis}.
The input to the CEGIS algorithm is a specification~$\varphi$ 
(possibly in multiple pieces -- say, as a temporal formula and a
language of possible solutions~\cite{sygus}).  
In the basic CEGIS loop, the synthesizer proposes a 
candidate solution~$S$, which is then checked against~$\varphi$.  
If it is correct, the CEGIS loop terminates; if not, a counterexample is 
provided and the synthesizer uses it to improve~$S$.
In practice, the CEGIS loop often faces performance issues, in particular, 
it can suffer from regressions: new candidate solutions may introduce 
errors that were not present in previous candidate solutions.
% (or in the original program). 
We address this issue by {\em making use of positive examples} (good 
traces) in addition to counterexamples (bad traces). The good traces
are used to learn constraints that ensure that these good traces are
preserved in the candidate solution programs proposed by the CEGIS loop. 
The PACES approach applies in many program synthesis contexts, but in
this paper, we focus on program repair for concurrency.

\noindent
{\bf Related work.}
The closest related work is by von Essen and Jobstmann~\cite{barbara},
which continues the work on program
repair~\cite{JGB05,GBC06,SDE08}. In~\cite{barbara}, 
the goal is to repair reactive systems (given as automata) according to an
LTL specification, with a guarantee that good traces do not disappear
as a result of the repair. Their algorithm is based on the
classic synthesis algorithm which translates the LTL specification to
an automaton. In contrast, we focus on the repair of concurrent
programs, and our algorithm uses positive examples and
counterexamples. 
% in the PACES loop.  

There are several recent algorithms for inserting synchronization by locks,
fences, atomic
sections, and other synchronization primitives~(\cite{Vechev:2010:ASS:1706299.1706338,Cherem:2008:ILA:1375581.1375619,ramalingam,SLJB08}).
%Most of these works do not consider bug fixing by instruction
%reordering.
Deshmukh et al.\ \cite{ramalingam} is the only one of these which
uses information about the correct parts of the program in bug fixing -- 
a proof of sequential correctness is used to identify positions for
locks in a concurrent library that is sequentially correct. 
CFix (Jin et al.\ \cite{shanlu}) can detect and fix
concurrency bugs using specific bug detection patterns and a fixing
strategy for each pattern of bug. 
Our approach relies on a general-purpose model checker and does not use
any patterns.

Our algorithm for fixing bad traces starts by generalizing
counterexample traces.  
In verification (as opposed to synthesis), 
concurrent trace generalization was used by
Sinha et
al.\ \cite{Sinha:2011:IA:1926385.1926433,DBLP:conf/sigsoft/SinhaW10};
and by Alglave et
al.\ \cite{DBLP:conf/cav/AlglaveKT13} for detecting errors due to weak
memory models.
Generalizations of good traces was previously used by Farzan et
al.\ \cite{Farzan:2013:IDF:2480359.2429086}, who create an inductive
data-flow graph (iDFG) to represent a proof of program correctness. 
They do not attempt to use iDFGs in synthesis. 

We use the model checker CBMC~\cite{cbmc} to generate both good and bad
traces.
Sen introduced concurrent directed random testing
\cite{Sen:2008:RDR:1375581.1375584}, which can be used to obtain good
or bad traces much faster than a model checker.
For a 30k LOC program their tool needs only about 2
seconds. 
We could use this tool to initially obtain good and bad traces
faster, thus increasing the scalability of our tool. 

% From last year's review:
% By the way, the idea of generalization of traces has been used extensive in 
% bug finding before. There is a paper in POPL 2013 (by Farzan et al) that does 
% this specifically for concurrency. Also, the idea of looking for cycles in 
% partial orders to detect violations of atomicity goes back to the conflict 
% serializability detection algorithm introduced by Papadimitriou, and has since 
% been used (in different forms) in a lot of papers. For example, you can replace 
% Algorithm 2 by a variant of an algorithm that uses the logical encoding (of 
% consistency conditions) by Sinha et al. You do blur the lines a little bit when 
% it comes to stating these as novel parts of your work.

% I think he means Algorithm 1

\noindent
{\bf Illustrative example.}
%\label{sec:illustrative_example}
We motivate our approach on the program $P$ in
Figure~\ref{fig:example}.    
There is a bug witnessed by the following trace:
$\trace_1 = A \to B \to 1 \to 2 \to 3$ (the assertion at line $3$ fails). 
Let us attempt to fix the bug using the algorithm from~\cite{cav2013}.
The algorithm discovers possible fixes by first generalizing the trace
into a partial order (Figure~\ref{fig:partialorder}, without the dotted 
edges) representing the happens-before relations necessary for the bug
to occur, and second, trying to create a cycle in the partial order to
eliminate the generalized counterexample.
It finds three possible ways to do this: swapping $B$ and
$C$, or moving $C$ before $A$, or moving $A$ after $C$, indicated by
the dotted edges in Figure~\ref{fig:partialorder}. 
Assume that we continue with swapping $B$ and $C$ to obtain program
$P_1$ where the first thread is $A; C; B$. 
Program $P_1$ contains an error trace $\trace_2 = A \to C \to
n \to p$ (the assertion at line $p$ fails). 
This bug was not in the original program, but was introduced
by our fix. 
We refer to this type of bug as a regression. 

%We see it was a premature to choose a fix without having information
%about the good parts of the program $P$. 
In order to prevent regressions, the algorithm learns from good
traces. 
Consider the following good trace $\trace_3 = A \to  B \to  C \to  1 \to 
2 \to  n \to  3 \to  p$.
The algorithm analyses the trace, and produces the graph in 
Figure~\ref{fig:colour1}. 
Here, the thick red edges indicate the reads-from
relation for {\tt assert} commands, and the dashed blue edges indicate
the reads-from relation for {\tt await} commands. 
Intuitively, the algorithm now analyses why the assertion at line $p$
holds in the given trace. 
This assertion reads the value written in line $B$ (indicated by the
thick red edge).
The algorithm finds a path from $B$ to $p$ composed entirely from
intra-thread sequential edges ($B \to C$ and $n \to p$) and dashed blue
edges ($C \to n$).
This path guarantees that this trace cannot be changed by different
scheduler choices into a path where $p$ reads from elsewhere and fails. 
From the good trace $\trace_2$ we thus find that there could be a
regression unless $B$ precedes $C$ and $n$ precedes $p$. 
Having learned this constraint, the synthesizer can find a better way
to fix $\trace_1$. 
Of the three options described above, it chooses the only way which does not
reorder $B$ and $C$, i.e., it moves $A$ after $C$. 
This fixes the program without regressions.

\begin{figure}[tb]
\ffigbox
{%
  \begin{subfloatrow}[3]
    \ffigbox[\FBwidth]% Width of subfloat
    {%
    \begin{minipage}{6cm}
      \small{ \tt \setlength{\tabcolsep}{2pt}
    init: x = 0; y = 0; z = 0\\
    \begin{tabular}{l l l}
    \underline{thread1}&\underline{thread2}&\underline{thread3}\\
      1: await(x==1)  & A: x:=1 & n: await(z==1)\\
      2: await(y==1)  & B: y:=1 & p: assert(y==1)\\
      3: assert(z==1)  & C: z:=1 \\
      \end{tabular}
      }
      \end{minipage}
    }
    {%
      \caption{Program $P$}\label{fig:example}%
    }
    \ffigbox[\FBwidth+0.3cm]% Width of subfloat
    {%
    \begin{tikzpicture}
      \tikzstyle{every state}=[draw=black, minimum size=12]
      \node[state]             (1) {{\tt 1}};
      \node[state,below of=1,yshift=0.25cm]  (2) {{\tt 2}};
      \node[state,below of=2,yshift=0.25cm]  (3) {{\tt 3}};
      \node[state,left of=1]  (a) {{\tt A}};
      \node[state,below of=a,yshift=0.25cm]  (b) {{\tt B}};
      \node[state,below of=b,yshift=0.25cm]  (c) {{\tt C}};

      \path[->]
            (b) edge (2)
            (1) edge (2)
            (2) edge (3)
            (3) edge (c)
            (a) edge (1)
            (c) edge[dotted] (b)
            (c) edge[in=200,out=160,dotted] (a)
      ;

    \end{tikzpicture}
    }
    {%
      \subcaption{Reorderings from bad traces}\label{fig:partialorder}%
    }
    
    \ffigbox[\FBwidth]% Width of subfloat
    {%
		\begin{tikzpicture}
      \tikzstyle{every state}=[draw=black, minimum size=12]
			\node[state]             (1) {{\tt 1}};
			\node[state,below of=1,yshift=0.25cm]  (2) {{\tt 2}};
			\node[state,below of=2,yshift=0.25cm]  (3) {{\tt 3}};
			\node[state,right of=1,xshift=0.15cm]  (a) {{\tt A}};
			\node[state,below of=a,yshift=0.25cm]  (b) {{\tt B}};
			\node[state,below of=b,yshift=0.25cm]  (c) {{\tt C}};
			\node[state,right of=a,xshift=0.15cm]  (n) {{\tt n}};
			\node[state,below of=n,yshift=0.25cm]  (p) {{\tt p}};

			\path[->]
						%(2) edge  (3)
						%(1) edge  (2)
						(n) edge  (p)
						(c) edge[red,thick] (3)
						(b) edge[red,thick] (p)
						(a) edge[blue,dashed] (1)
						(b) edge[blue,dashed] (2)
						(c) edge[blue,dashed] (n)
						(b) edge (c)
						%(c) edge[in=240,out=120,dashed] (a)
			;
    \end{tikzpicture}
    }
    {%
			\subcaption{Learning from a good trace}\label{fig:colour1}%
		}
    
  \end{subfloatrow}

}
{
  \vspace{-2ex}
  \caption{Program analysis with good and bad
traces}\label{fig:simple_example}}
\end{figure}
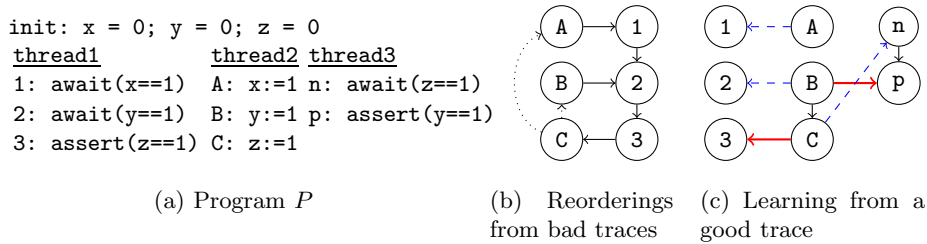

%Our technique applies to real concurrency bugs as found in the Linux kernel.
%We explain the technique on several examples of simplified Linux
%drivers in sec%tion~\ref{sec:impl}.  

\section{Programming Model and the Problem Statement}
\label{sec:defs}

% \subsection{Programming Model}

% \paragraph{Program Variables.}
% \paragraph{A Concurrent While-language.}
% Fix $\Vars$ to be a finite set of variables.
% , and let $\Vars'$ contain
% the primed versions of variables in $\Vars$.
% A {\em valuation} $\Valuation$ of $\Vars$ is a map from variables in
% $\Vars$ to their values.
% Following standard procedure (see for example,~\cite{}), we model
% program statements as first-order formulae over $\Vars \cup \Vars'$
% where the unprimed and primed versions of a variable represent its
% values before and after execution of the statement.
% For example, an assignment {\tt x := x + y} would be represented as the
% assertion $x' = x + y$.
% We assume that there is a special variable $\errVar \in \Vars$ whose
% value indicates whether or not a failure has occurred.
% -- it is $0$ (resp. $1$) if an error has not (resp. has) occurred.
%
% Low-level kernel code and device drivers usually incorporate aspects of
% both shared-memory and message-passing concurrency paradigms.
% Low-level device drivers use many complex concurrency constructs.
% Here, we use a simplified shared-memory programming model with just $3$
% constructs -- locks, atomic sections and wait/notify semaphores.
Our programs are composed of a fixed number (say $n$) threads
written in the {\sc Cwhile} language (Figure~\ref{fig:syntax}).
Each statement has a unique program location and each thread has unique
initial and final program locations.
% The constructs in {\sc Cwhile} are standard.
Further, we assume that execution does not stop on assertion
failure, but instead, a variable $\errVar$ is set to $1$.
The {\tt await} construct is a blocking assume, i.e., execution of
{\tt await(cond)} stops till {\tt cond} holds.
For example, a lock construct can be modelled as {\tt
atomic \{  await(lock\_var == 0); lock\_var := 1 \}}.
Note that {\tt await} is the only blocking operation in {\sc Cwhile} --
hence, we call the {\tt await} operations {\em preemption-points}.
\begin{figure} 
  \vspace{1ex}
\begin{alltt}
iexp ::= iexp + iexp | iexp / iexp | iexp * iexp | var | constant
bexp ::= iexp >= iexp | iexp == iexp | bexp && bexp | !bexp
stmt ::= variable := iexp | variable := bexp | stmt; stmt | assume(bexp) 
         | if (*) stmt else stmt | while (*) stmt | atomic \{ stmt \}
         | assert(bexp) | await(bexp)
thrd ::= stmt                               prog  ::= thrd | prog\(\parallel\)thrd
\end{alltt}
\vspace{-4ex}
\caption{Syntax of programming language}
\label{fig:syntax}
\end{figure}

\!\!\!\!\!\!\!\!
\paragraph{Semantics.}
The {\em program-state} $\ProgState$ of a program $\Prog$ is given by
$(\Valuation, (l^1, \ldots, l^n))$ where $\Valuation$ is a
valuation of variables, and each $l^t$ is a thread~$t$~program location.
Execution of the thread~$t$ statement at location $l^t$ is
represented as $\ProgState l^t \ProgState'$ where
$\ProgState = (\Valuation, (\ldots, l^t, \ldots))$ and
$\ProgState' = (\Valuation', (\ldots, l^{t'}, \ldots))$,
and $l^{t'}$ and $\Valuation'$ are the program location and
variable valuation after executing the statement from $\Valuation$. 
A {\em trace} $\trace$ of $\Prog$ is a sequence $\ProgState_0 l_0
 \ldots \ProgState_m$ where 
\begin{inparaenum}[(a)]
\item $\ProgState_0 = (\Valuation, (l^1_\iota, \ldots,
  l^n_\iota))$ where each $l^t_\iota$ is the initial location of
  thread $t$; and
\item each $\ProgState_i l_i \ProgState_{i+1}$ is a thread~$t$
  transition for some $t$.
% \item if $l_i$ is an {\tt assume(b)} or {\tt await(b)} and $\ProgState_i
% = (\Valuation_i, (l^1_i, \ldots, l^n_i))$, then $b$ evaluates to true
% under the valuation $\Valuation_i$.
% \arsays{ Not required as that is the semantic}
% \item there is no transition $\ProgState_m l_m \ProgState_{m+1}$ for
  % any $l_m$ and $\ProgState_{m+1}$.
\end{inparaenum}
Trace $\trace$ is {\em complete} if $\ProgState_m =
(\Valuation_m, (l_f^1, \ldots, l_f^k))$, where each $l_f^t$ is the
final location of thread~$t$.
We say $\ProgState_i l_i \ldots \ProgState_n$ is {\em equal modulo
error-flag} to $\ProgState_i' l_i \ldots \ProgState_n'$ if each
$\ProgState_k$ and $\ProgState_k'$ differ only in the valuation of the
variable $\errVar$.
%
% \arsays{Wrong place for these following definitions}
% {\new 
% We write $i<j$ to indicate that instruction $i$ occurs before $j$ in
% the trace.
% By $\ProgState[v]$ we denote the value of variable $v$ in $\ProgState$.
% We say $\ProgState_x=\ProgState_y$ if $\forall v.\ \ProgState_x[v]=\ProgState_y[v]$.
% }

% A {\em trace} of $\Prog$ is a sequence $(\Valuation^0, l^0_1,
% \ldots, l^0_n)(\Valuation^1, l^1_1, \ldots, l^1_n)(\Valuation^2,
% l^2_1, \ldots, l^2_n)\ldots$ where for each $i \in
% \mathbb{Z}^{\geq 0}$, there exists a $\thread(i) \in \{ \CFG_1, 
% \ldots, \CFG_n \}$ such that $((\Valuation^i, l^i_1, \ldots, l^i_n),
% (\Valuation^{i+1}, l^{i+1}_1, \ldots, l^{i+1}_n))$ is a
% thread-$\thread(i)$ transition. 
% A trace is error-free if the special variable $\Valuation^i(\errVar)
% \neq 1$ for all $i \in \mathbb{Z}^{\geq 0}$.

% A trace is {\em preemption-free} if for each $i \in \mathbb{Z}^{\geq
% 0}$ if $((\Valuation^i, l^i_1, \ldots, l^i_n), (\Valuation^{i+1},
% l^{i+1}_1, \ldots, l^{i+1}_n))$ is a $\CFG_i$ transition, and
% $((\Valuation^{i+1}, l^{i+1}_1, \ldots, l^{i+1}_n), (\Valuation^{i+2},
% l^{i+2}_1, \ldots, l^{i+2}_n))$ is a $\CFG_j$ transition with
% $i \neq j$, then there exists no $(\Valuation',
% l'_1, \ldots, l'_n)$ such that $((\Valuation^i, l^i_1, \ldots,
% l^i_n), (\Valuation', l'_1, \ldots, l'_n))$ is a $\CFG_i$
% transition.
Trace $\trace$ is {\em preemption-free} if every context-switch occurs
either at a preemption-point ({\tt await} statement) or at the end of a
thread's execution, i.e., if where $\ProgState_i l_i
\ProgState_{i+1}$ and $\ProgState_{i+1} l_{i+1} \ProgState_{i+2}$ are
transitions of different threads (say threads $t$ and $t'$), either the
next thread~$t$ instruction after $l_i$ is an {\tt await}, or the
thread~$t$ is in the final location in $\ProgState_{i+1}$.
Similarly, we call a trace {\em sequential} if every context-switch
happens at the end of a thread's execution.

A trace $\trace = \ProgState_0 l_0 \ldots \ProgState_m$
is {\em bad} if the error variable $\errVar$ has values $0$ and $1$ in
$\ProgState_0$ and $\ProgState_m$, respectively; otherwise, $\trace$ is
{\em good} trace.
We assume that the bugs present in the input
programs are {\em data-independent} -- if $\trace = \ProgState_0 l_0
\ProgState_1 \ldots \ProgState_n $ is bad, so is every trace $\trace'
=\ProgState'_0 l'_0 \ProgState'_1 \ldots \ProgState'_n $ where
$l_i=l'_i$ for all $0 \leq i < n$.

\noindent{\bf Program transformations and Program constraints.}
We consider two kinds of transformations for fixing bugs:
\begin{compactitem}
\item~A {\em reordering transformation} $\reorder = \switch{l_1}{l_2}$
  transforms $\Prog$ to $\Prog'$ if location $l_1$ immediately precedes
  $l_2$ in $\Prog$ and $l_2$ immediately precedes $l_1$ in
  $\Prog'$.
  We only consider cases where the sequential semantics
  are preserved, i.e., if
  % a reordering $\switch{l_1}{l_2}$ is allowed if and only if
  \begin{inparaenum}[(a)]
  \item~$l_1$ and $l_2$ are from the same basic block; and
  \item~$l_1; l_2$ is equivalent to $l_2; l_1$.
  \end{inparaenum}
\item~An {\em atomic section transformation} $\reorder = [l_1;l_2]$
  transforms $\Prog$ to $\Prog'$ if 
  % location $l_1$ is immediately before location $l_2$ and
  neighbouring locations $l_1$ and $l_2$ are in an atomic
  section in $\Prog'$, but not in $\Prog$.
\end{compactitem}
We write
$\Prog\xrightarrow{\reorder_1\ldots \reorder_k}\Prog'$ if applying each
of $\reorder_i$ in order transforms $\Prog$ to $\Prog'$.
We say transformation $\reorder$ {\em
acts across preemption-points} if either $\reorder = \switch{l_1}{l_2}$
and one of $l_1$ or $l_2$ is a preemption-point; or if $\reorder = [l_1;
l_2]$ and $l_2$ is a preemption-point.
%
% Further, 
% extending the terminology,

Given a program $\Prog$, we define {\em program constraints} to
represent sets of programs that can be obtained through
applying program transformations on $\Prog$.
% We consider the following types of program constraints:
\begin{compactitem}
\item~{\em Atomicity constraint}: Program $\Prog'\!\models [l_i; l_j]$ if
  $l_i$ and $l_j$ are in an atomic block.
\item~{\em Ordering constraint}: Program $\Prog' \models l_i \leq
  l_j$ if $l_i$ and $l_j$ are from the same basic block and either
  $l_i$ occurs before $l_j$, or $\Prog'$ satisfies $[ l_i; l_j]$.
\end{compactitem}
If $\Prog' \models \ProgCons$, we say that $\Prog'$ {\em satisfies}
$\ProgCons$.
Further, we define conjunction of $\ProgCons_1$ and
$\ProgCons_2$ by letting $\Prog' \models \ProgCons_1 \wedge \ProgCons_2
\Leftrightarrow \left( \Prog' \models \ProgCons_1 \wedge \Prog' \models
\ProgCons_2 \right)$.

\noindent{\bf Trace Transformations and Regressions.}
%
% \paragraph{Free Trace Transformations.}
A trace $\trace = \ProgState_0 l_0  \ldots \ProgState_m$
{\em transforms} into a trace $\trace' = \ProgState_0' l_0'
 \ldots \ProgState_m'$ by {\em switching} if:
\begin{inparaenum}[(a)]
\item $\ProgState_0 l_0 \ldots \ProgState_n = \ProgState_0' l_0' \ldots
  \ProgState_n'$ and the suffixes
  $\ProgState_{n+2} l_{n+2} \ldots \ProgState_m$ and $\ProgState_{n+2}'
  l_{n+2}' \ldots \ProgState_m'$ are equal modulo error-flag; and
\item $l_n = l_{n+1}' \wedge l_{n+1} = l_n'$.
\end{inparaenum}
% Intuitively, $\trace'$ is obtained from $\trace$ by switching 
% statements $l_n$ and $l_{n+1}$ without changing the rest of the trace.
We label switching transformations as a:
\begin{compactitem}
\item {\em Free transformation} if $l_n$ and $l_{n+1}$ are from
  different threads. 
  We write $\trace'\in f(\trace)$ if a sequence of free transformations
  takes $\trace$ to $\trace'$.
\item {\em Reordering transformation $\reorder =
    \switch{l^\sharp}{l^\flat}$ acting on $\trace$} if $l_n = l^\sharp$
    and $l_{n+1} = l^\flat$.
  % \arsays{clarify. put repeated applications}
  We have $\trace' \in \reorder(\trace)$ if repeated applications of
  $\reorder$ transformations acting on $\trace$ give $\trace'$. 
  Similarly, $\trace'\in \reorder^f(\trace)$ if repeated applications of
  $\reorder$ and free transformations acting on $\trace$
  give $\trace'$.
\end{compactitem}
Similarly, $\trace'$ is obtained by {\em atomicity
transformation $\reorder = [l_1, l_2]$ acting on a trace} $\trace$
if $\trace' \in f(\trace)$, and there are no context-switches between
$l_1$ and $l_2$ in $\trace'$.

\paragraph{Trace analysis graphs.}
We use trace analysis graphs to characterize data-flow and
scheduling in a trace.
First, given a trace $\trace = \ProgState_0 l_0 
\ldots$, we define
% \begin{compactitem}
% \item~The functions $\readn(i)$ and $\writen(i)$ returns a set of
  % variables the statement at location $l_i$ reads and writes,
  % respectively;
% \item~the functions $\before(i)$ returns all command before $i$ in the
  % trace, similar for $\after(i)$ (\arsays{This should not be needed}).
% \item~If $v \in \writen(i)$, the function $\written(i,v)$ returns the
  % value the statement at location $l_i$ wrote to variable $v$; 
% \item~The functions $\mathit{awaits}$, $\mathit{assumes}$, and
  % $\mathit{asserts}$ denote the set of all locations in the trace that 
  % correspond to {\tt await} statements, {\tt assume} statements, and
  % {\tt assert} statements respectively.
% \item~The function $\last(i,v)$ returns the location of the latest
  % statement before $l_i$ that writes to variable $v$;
% \item~
the function $\depends$ to recursively find the data-flow
edges into the $l_i$.
Formally, $\depends(i) = \cup_{v}\{ (\last(i,v), i) \} \cup
\depends(\last(i,v))$ where $v$ ranges over variables read by $l_i$, and
$\last(i,v)$ returns $j$ if $l_i$ reads the value of $v$ written by
$l_j$ and $\last(i,v) = \bot$ if no such $j$ exists.
As the base case, we define $\depends(\bot) = \emptyset$.
  % in the trace. 
  % the position of the
  % last write to $v$ before $l_i$.
% \end{compactitem}

Now, a {\em trace analysis graph} for trace $\trace =
\ProgState_0l_0\ldots \ProgState_n$ is a
multi-graph  $G(\trace) = \langle V, \to \rangle$, where $V = \{ \bot \}
\cup \{ i | 0 \leq i \leq n \}$ are the positions in the trace along
with $\bot$ (representing the initial state) and $\to$ contains the
following types of edges.
% (we
% label the types with colours for convenience).
% followi $E$ are reads-from and sequential
% thread-order edges labelled with a colour. 
% The edges are labelled with a colour. 
% We denote $(x\to y)$ as an edge in
% the graph between two locations $x$ and $y$.
\begin{compactenum}
\item~{\em Intra-thread order} ($\blackn$): We have $x \to y$ if either
  $x < y$, and $l_x$ and $l_y$ are from the same thread, or if $x =
  \bot$.
\item~{\em Data-flow into conditionals} ($\bluen$): We have
  $\bigcup_{a\in\mathit{conds}} \depends(a) \subseteq \to$ where $x \in
  conds$ iff $l_x$ is an assume or an await statement.
\item~{\em Data-flow into assertions} ($\redn$): We have
  $\bigcup_{a\in\mathrm{asserts}} \depends(a) \subseteq \to$ where
  $x \in asserts$ iff $l_x$ is an assert statement.
\item {\em Non-free order} ($\pinkn$): We have $x \to y$ if $l_x$ and
  $l_y$ write two different values to the same variable. 
  Intuitively, the non-free orders prevent switching transformations
  that switch $l_x$ and $l_y$.
\end{compactenum}
% Given a trace $\trace$, we define $\bluen(\trace)$ to be the set of blue
% edges in the trace analysis graph $G(\trace)$.

\paragraph{Regressions.}
Suppose $\Prog\xrightarrow{\reorder_1,\ldots,\reorder_k}\Prog'$.
We say $\reorder_1,\ldots,\reorder_k$ introduces a {\em
regression} with respect to a good trace $\trace = \ProgState_0 l_0
\ldots \ProgState_m$ of $\Prog$ if there exists a trace $\trace' = \ProgState_0' l_0'
\ldots  \ProgState_m' \in \reorder^f_k\circ\ldots\circ
\reorder^f_1(\trace)$ such that:
\begin{inparaenum}[(a)]
\item $\trace'$ is a bad trace of $\Prog'$;
\item $\trace$ does not freely transform into any bad trace of $\Prog$; and
\item for every data-flow into conditionals edge $x \to y$ (say $l_y$
  reads the variables $\Vars$ from $l_x$) in $G(\trace)$, the edge $p(x)
  \to p(y)$ is a data-flow into conditionals edge in $G(\trace')$ (where
  $l'_{p(y)}$ reads the same variables $\Vars$ from $l'_{p(x)}$).
  Here, $p(i)$ is the position in $\trace'$ of instruction at position
  $i$ in $\trace$ after the sequence of switching transformations that
  take $\trace$ to $\trace'$.
\end{inparaenum}
We say $\reorder_1\ldots\reorder_k$ introduces a regression with respect
to a set $T_G$ of good traces if it introduces a
regression with respect to at least one trace $\trace \in T_G$.

Intuitively, a program-transformation induces a regression if it allows
a good trace $\trace$ to become a bad trace $\trace'$ due to the program
transformations.
Further, we require that $\trace$ and $\trace'$ have the conditionals
enabled in the same way, i.e., the {\tt assume} and {\tt await}
statements read from the same locations.
% Note that atomic section transformations cannot introduce regressions by
% themselves.
\begin{remark}
  \label{rem:regressions}
  The above definition of regression attempts to capture the intuition
  that a good trace transforms into a ``similar'' bad trace.
  The notion of similar asks that the traces have the same data-flow
  into conditionals -- this condition can be relaxed to obtain more
  general notions of regression.
  However, this makes trace analysis and finding regression-free fixes
  much harder (See Example~\ref{ex:non_locality}).
\end{remark}

\begin{example}
  In Figure~\ref{fig:simple_example}, the trace $\trace = A ; B
  ; C ; n ; p$ transforms under $\switch{B}{C}$ to $\trace' = A
  ; C ; B ; n ; p$, which freely transforms to $\trace'' = A ;
  C ; n ; p ; B$.
  Hence, $\switch{B}{C}$ introduces a regression with respect to $\trace$
  as $\trace$ does not freely transform into a bad trace,
  and $\trace'$ is bad while the {\tt await} in $n$ still reads from $C$.
\end{example}

\noindent{\bf The Regression-free Program-Repair Problem.}
% We are now ready to state the main program-repair problems for
% regression-free synthesis.
Intuitively, the program-repair problem asks for a correct program
$\Prog'$ that is a transformation of $\Prog$.
Further, $\Prog'$ should preserve all sequential behaviour
of $\Prog$; and if all preemption-free behaviour of $\Prog$ is good, we
require that $\Prog'$ preserves it.

\paragraph{Program repair problem.}
The input is a program $\Prog$ where all complete sequential traces are
good.
The result is a sequence of program transformations $\reorder_1\ldots
\reorder_n$ and $\Prog'$, such that
\begin{inparaenum}[(a)]
\item~$\Prog\xrightarrow{\reorder_1\ldots \reorder_n}\Prog'$;
\item~$\Prog'$ has no bad traces;
\item~for each complete sequential trace $\trace$ of $\Prog$, there
  exists a complete sequential trace $\trace'$ of $\Prog'$ such that
  $\trace' \in \reorder_1 \circ \reorder_2 \ldots \circ \reorder_n
  (\trace)$; and
\item~if all complete preemption-free traces of $\Prog$ are good, then
  for each such trace $\trace$, there exists a complete preemption-free
  trace $\trace'$ of $\Prog'$ such that $\trace' \in \reorder_1 \circ
  \reorder_2 \ldots \circ \reorder_n (\trace)$.
\end{inparaenum}
We call the conditions (c)~and (d)~the {\em preservation of sequential and
correct preemption-free behaviour}.
% Note that we do not allow free transformations in the previous
% definition, i.e., we use $\reorder_i$ instead of $\reorder_i^f$.

\paragraph{Regression-free error fix.}
Our approach to the above problem is through repeated
regression-free error fixing.
% Intuitively, the regression-free error fixing problem asks to fix one
% bug without introducing regressions.
Formally, the regression-free error fix problem takes a set of good
traces $T_G$, a program $\Prog$ and a bad trace $\trace$ as input,
and produces transformations $\reorder_1,\ldots,\reorder_k$ and $\Prog'$
such that $\Prog\xrightarrow{\reorder_1\ldots\reorder_k}\Prog'$,
$\trace'\in\reorder^f_k\circ\ldots\circ\reorder^f_1(\trace)$ is a trace
in $\Prog'$, and $\reorder_1,\ldots,\reorder_k$ does not introduce a {\em
regression} with respect to $T_G$.

\section{Good and Bad Traces}

Our approach to program-repair is through
learning regression preventing constraints from good traces 
and error eliminating constraints from bad traces.
% From good and bad traces, we learn constraints to prevent regressions
% and fix errors respectively.
% We show how constraints are learned from good traces to prevent
% regressions and how fixes are found for existing bad traces.
% In Section~\ref{sec:algo}, we then combine the constraints and fixes to
% obtain a regression-free solution program.

\subsection{Learning from Good Traces}
\label{sec:learn_good}

% Here, the aim is to learn regression-preventing constraints from good
% traces.
% \comment{For now, we assume that each good trace contains only one assertion.
% Multiple assertions can be handled by running the
% learning algorithm once for each assertion.}
%
Given a trace $\trace$ of $\Prog$, a program constraint $\ProgCons$ is a
{\em sound regression preventing constraint} for $\trace$ if every
sequence of program transformations $\reorder_1,\ldots,\reorder_k$,
such that $\Prog\xrightarrow{\reorder_1\ldots
\reorder_k}\Prog'$ and $\Prog'\models\ProgCons$, does not introduce a
regression with respect to $\trace$.
Further, if every $\reorder_1\ldots\reorder_k$, such that
$\Prog\xrightarrow{\reorder_1\ldots \reorder_k}\Prog'$ and
$\Prog'\not\models\ProgCons$, introduces a regression with respect to
$\trace$, then $\ProgCons$ is a {\em complete regression preventing
constraint}.

\begin{example}
  \label{ex:sound_complete_cons}
  Let the program $\Prog$ be ${\tt \{ 1: x := 1; 2: y := 1 \} || \{ A:
  await (y = 1);}$ ${\tt B:assert (x = 1) \}}$. 
  In Figure~\ref{fig:orig_trace}, the constraint $\ProgCons^* = (1 <
  2 \wedge A<B)$ is a sound and complete regression-preventing
  constraint for the trace $1 \to 2 \to A \to B$.
\end{example}

\begin{figure}[b]
\ffigbox
{%
\scriptsize%
  \begin{subfloatrow}[3]
    \ffigbox[\FBwidth]% Width of subfloat
    {%
      \begin{tikzpicture}
        \node[rounded rectangle] (c) {{\tt 1: x:=1}};
        \node[rounded rectangle, below of=c, yshift=10] (d) {{\tt 2: y:=1}};
        \node[rounded rectangle, left of=c, xshift=-20, yshift=-40] (a) {{\tt A: await(y==1)}};
        \node[rounded rectangle, below of=a, yshift=10] (b) {{\tt B: assert(x==1)}};

        \draw[->] (a) to (b);
        \draw[->] (c) to (d);
        \draw[->, blue, dashed] (d) to (a);
        \draw[->, red,thick] (c.south west) to (b.north east);
      \end{tikzpicture}
    }
    {\vspace{-2ex}\subcaption{}\label{fig:orig_trace}}
    \ffigbox[\FBwidth]% Width of subfloat
    {%
      \begin{tikzpicture}
        \node[rounded rectangle] (c1) {{\tt 1: x:=1}};
        \node[rounded rectangle, below of=c1, yshift=10] (d1) {{\tt 2: y:=1}};
        \node[rounded rectangle, left of=c1, xshift=-30, yshift=-20] (a1) {{\tt A: await(y==1)}};
        \node[rounded rectangle, below of=a1, yshift=10] (b1) {{\tt B: assert(x==1)}};
        \node[rounded rectangle, below of=b1, yshift=10] (b2) {{\tt C: a:=1}};
        \node[rounded rectangle, below of=d1, yshift=-15] (e2) {{\tt 3: assume(a==1)}};
        \node[rounded rectangle, below of=e2, yshift=10] (e1) {{\tt 4: x:=0}};

        \draw[->] (a1) to (b1);
        \draw[->] (b1) to (b2);
        \draw[->] (c1) to (d1);
        \draw[->] (e2) to (e1);
        \draw[->, blue, dashed] (d1) to (a1);
        \draw[->, red,thick] (c1.south west) to (b1.north east);
        \draw[->, green, thin] (b1.340) to (e1.west);
        \draw[->, blue, dashed] (b2.east) to (e2.north west);
      \end{tikzpicture}
    }
    {\vspace{-2ex}\subcaption{} \label{fig:from_read_trace}}
    \ffigbox[\FBwidth-0.5cm]% Width of subfloat
    {%
      \begin{tikzpicture}
        \node[rounded rectangle] (c1) {{\tt 1: x:=1}};
        \node[rounded rectangle, below of=c1, yshift=10] (e1) {{\tt 2': y:=2}};
        \node[rounded rectangle, below of=e1, yshift=10] (d1) {{\tt 2: y:=1}};
        \node[rounded rectangle, left of=e1, xshift=-30, yshift=-30] (a1) {{\tt A: await(y>=1)}};
        \node[rounded rectangle, below of=a1, yshift=10] (b1) {{\tt B: assert(x==1)}};

        \draw[->] (a1) to (b1);
        \draw[->, blue, dashed] (d1) to (a1);
        \draw[->, red,thick] (c1.west) to (b1.east);
        \draw[->] (c1.south east) to[bend left=60] (d1.north east);
      \end{tikzpicture}
    }
    {\vspace{-2ex}\subcaption{} \label{fig:write_order_trace}}
    \end{subfloatrow}
}
{
  \vspace{-2ex}
\caption{Sample Good Traces for Regression-preventing constraints}\label{fig:traces}
}
\end{figure}
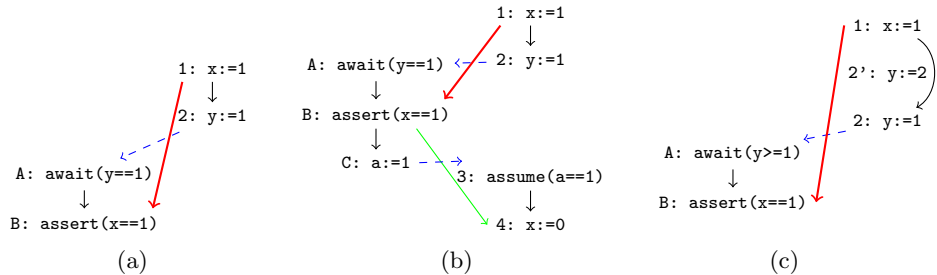

% The following lemma states that the sound and complete regression
% preventing constraint can be computed in time exponential in the length
% of the given trace.
\begin{lemma}
  \label{lem:learning_completeness}
  For a program $\Prog$ and a good trace $\trace$, the sound and
  complete regression-preventing constraint $\ProgCons^*$ is computable
  in exponential time in $\vert \trace \vert$.
\end{lemma}
Intuitively, the proof relies on an algorithm 
% (which we call $\LearnGoodGen$)
that iteratively applies all possible free
and program transformations in different combinations (there are a
finite, though exponential, number of these) to $\trace$.
It then records the constraints satisfied by programs obtained by
transformations that do not introduce regressions. 

\comment{
While the complexity is exponential, we can
show that this cost is unavoidable.
We do not present the proof here, but only state that
% , but instead refer the reader to the appendix.
it is non-constructive and is based on Shannon's lower bounds on
circuit complexity for boolean functions.
\arsays{This is a bit shady: is there a family of boolean functions that
has poly size circuits/poly size turing machines, but a exponential size
formulae}
\begin{lemma}
  There exist a class of programs $\Prog_n$, and traces $\trace_n$ of
  length $O(n)$ such that the most-general regression-preventing
  constraint is of size $\Theta(\frac{2^n}{n})$.
\end{lemma}
}

% \ttsays{I am not convinced because the variables here are the happens-before relations in the formula. This means we have $O(n^2)$ rather than $O(n)$ variables.}

% \paragraph{Non-locality.}
% Although Lemma~\ref{lem:learning_completeness} give us a correct and
% optimal procedure for learning constraints, the constraints produced are
% usually large and impractical.
The sound and complete constraints are usually large and
impractical to compute.
Instead, we present an algorithm to compute sound
regression-preventing constraints.
%
% {\color{red}
% As explained in Section~\ref{sec:illustrative_example}, the main idea
% behind our constraint learning procedures is the concept of {\em
% protection} of critical reads.
% We say that an order $l_0 < l_1$ is protected by a constraint
% $\ProgCons \wedge \SchedCons$ in a trace $\trace$ of a program $\Prog$
% if: \arsays{formal statement}.
% Intuitively, an order $l_0$ before $l_1$ in a trace is protected by a
% constraint if no sequence of transformations respecting the constraint
% can reverse the order of $l_0$ and $l_1$ in the trace.
%
% However, the notion of protecting the critical reads to assertions is
% not sufficient to compute under- and over-approximations of the
% regression-preventing constraints.
The main issue here is non-locality, i.e., statements that are not close
to the assertion may influence the regression-preventing constraint.

% }
% with both
% over-approximation and under-approximation of regression-preventing
% constraints.

\begin{example}
  \label{ex:non_locality}
  The trace in Figures~\ref{fig:from_read_trace} is a simple extension
  of Figure~\ref{fig:orig_trace}.
  However, the constraint $(1 \leq 2 \wedge A \leq B)$
  (from Example~\ref{ex:sound_complete_cons}) does not prevent
  regressions for Figure~\ref{fig:from_read_trace}.
  An additional constraint $B \leq C \wedge 3 \leq 4$ is needed as
  reordering these statements can lead to the assertion failing by
  reading the value of {\tt x} ``too late'', i.e., from the statement
  {\tt 4} (trace: $1 \to 2 \to A\to C \to 3 \to 4 \to B$).
  
  Figure~\ref{fig:write_order_trace} clarifies our definition of
  regression, which requires that the data-flow edges into
  assumptions and awaits need to be preserved.
  The await can be activated by both {\tt 2} and {\tt 2'}; in the trace
  we analyse it is activated by {\tt 2}.
  Moving {\tt 2'} before {\tt 1} could activate the await ``too early''
  and the assertion would fail (trace: $2'\to A\to B$).
  However, it is not possible to learn this purely with data-flow
  analysis -- for example, if statement {\tt 2'} was {\tt y := -1}, then
  this would not lead to a bad trace.
  Hence, we exclude such cases from our definition of regressions by
  requiring that the await reads {\tt A} reads from the same location.
\end{example}

\comment{
\begin{example}[Problems for over-approximation]
  Consider the modified version of the trace in
  Figure~\ref{fig:over_approx_trace}.
  Here, the constraint $\ProgCons^* = (A < B \wedge 1 < 2)
  \vee \SchedBefore{1}{B}$ is too strong, i.e., there are programs
  that violate the constraint and yet do not introduce regressions.
  The problem here is that the assertion may succeed even if the current
  version of the data-flow into the assertion is not respected -- even
  if the assertion reads the value of {\tt x} too early, i.e., from {\tt
  0'}, the assertion still succeeds.
\end{example}
}

\paragraph{Learning Sound Regression-Preventing Constraints.}
The sound regression-preventing constraint learned by our algorithm for
a trace ensures that the data-flow into an assertion is preserved.
This is achieved through two steps: suppose an assertion at location
$l_a$ reads from a write at location $l_w$. 
First, the constraint ensures that $l_w$ always happens before
$l_a$.
Second, the constraint ensures that no other writes 
interfere with the above read-write relationship.
% We explain the two steps below.

For ensuring happens-before relationships, we use the notion of a {\em
cover}.
Intuitively, given a trace $\trace$ of $\Prog$ where location $l_x$
happens before location $l_y$,
% happens-before relation $x \to y$ is a program constraint $\ProgCons$
we learn a $\ProgCons$ that ensures that if $\Prog' \models \ProgCons$,
then each trace $\trace'$ of $\Prog'$ obtained as free and program
transformations acting on $\trace$ satisfies the happens-before
relationship between $l_x$ and $l_y$.
Formally, given a trace $\trace$ of program $\Prog$, we call a path $x_1
\to x_2 \to \ldots \to x_n$ in the trace analysis graph a {\em cover} of
edge $x\to y$ if $x = x_1 \wedge y = x_n$ and each of $x_i \to x_{i+1}$
is either a intra-thread order edge, or a  data-flow into conditionals
edge, or a non-free order edge.

Given a trace $\trace = \ProgState_0 l_0 \ProgState_1 l_1 \ldots
\ProgState_n$, where statement at position $r$ (i.e., $l_r$) reads a set
of variables (say $\Vars$) written by a statement at position $w$ (i.e.,
$l_w$), the the non-interference edges define a sufficient set of
happens-before relations to ensure that no other statements can
interfere with the read-write pair, i.e., that every other write to
$\Vars$
either happens before $w$ or after $r$.
Formally, we have that $\interfere(w \to r) = \{ r \to w' \mid w' > r
  \wedge \writen(l_{w'}) \cap \writen(l_w) \cap \readn(l_r) \neq
\emptyset \} \cup \{ w' \to w \mid w' < w
  \wedge \writen(l_{w'}) \cap \writen(l_w) \cap \readn(l_r) \neq
\emptyset \}$ where $\readn(l)$ and $\writen(l)$ are the 
variables read and written at location $l$.
If $w = \bot$, we have $\interfere(w \to r) = \{ r \to w' \mid w' > r
  \wedge \writen(l_{w'}) \cap \readn(l_r) \neq
\emptyset \}$.
\begin{algorithm}
  \begin{algorithmic}[1]
    \REQUIRE A good trace $\trace$
    \ENSURE Regression-preventing constraint $\ProgCons$
    \STATE $\ProgCons\gets true; G \gets G(\trace)$
    \FORALL {$e\in\left(\redn(G)\cup \bigcup_{f \in \redn(G)}\interfere(f)\right)$}
    \STATE \textbf{if}~$e$ is not
    covered~\textbf{then}~\textbf{return}~$\bigwedge \{ l_x \leq l_y
    \mid x \to y~\mbox{is a intra-thread order edge} \} $~\label{line:ret_fallback}
      \STATE $\ProgCons' \gets {\tt false}$
      \FORALL { $x_1 \to x_2 \to \ldots \to x_n$ cover of $e$}  \label{line:pick_cover}
      \STATE $\ProgCons' \gets \ProgCons' \vee \bigwedge \{ l_{x_i} \leq
	l_{x_{i+1}} \mid x_i \to x_{i+1}~\mbox{is a intra-thread order
      edge and $x_i \neq \bot$}$
      \item[] $\qquad\qquad l_{x_i}$ and $l_{x_{i+1}}$ are from the same
  execution of a basic block in $\trace$ \}
        % \FORALL {edge $x\to y\in \cover$}
          % \IF {$x\to y$ is black}
            % \STATE $\ProgCons'\gets\ProgCons'\wedge x\leq y$
          % \ENDIF
        % \ENDFOR
      \ENDFOR
      \STATE $\ProgCons \gets \ProgCons \wedge \ProgCons'$
    \ENDFOR
    \RETURN {$\ProgCons$} \label{line:ret_normal}
  \end{algorithmic}
  \caption{Algorithm $\LearnGoodUnder$} 
  \label{algo:learn_good}
\end{algorithm}

% The main idea of Algorithm~\ref{algo:learn_good} is to {\em protect} the
% data-flow into the assertions.
% the learned constraint may be too strong and may eliminate programs that
% are not necessarily regression inducing.
Algorithm~\ref{algo:learn_good} works by ensuring that for each data-flow
into assertions edge $e$, the edge itself is covered and that the
interference edges are covered.
% It maintains a global constraint $\ProgCons$ to be returned.
For each such cover, the set of intra-thread order edges needed
for the covering are conjuncted to obtain a constraint.
We take the disjunction $\ProgCons'$ of the constraints produced by all
covers of one edge and add it to a constraint $\ProgCons$ to
be returned.
If an edge cannot be covered, the algorithm falls back by returning a
constraint that fixes all current intra-thread orders.
The algorithm can be made to run in
polynomial time in $\vert \trace \vert$ using standard dynamic
programming techniques.
% only a polynomial time in the size of the trace.
\begin{theorem}
  \label{lem:learning_soundness}
  Given a trace $\trace$, Algorithm~\ref{algo:learn_good} returns a
  constraint $\ProgCons$ that is a sound regression-preventing
  constraint for $\trace$ and runs in polynomial time in $\vert \trace
  \vert$.
\end{theorem}
\begin{proof}[Outline]
  % \arsays{There is a "minor" problem here and in definition of
    % preservation of blue edges. The $x$ and $y$ in $x \to y$ do not
    % represent the same edges in $\trace$ and $\trace'$.
    % Introduce a function (say $\trace -- \trace'$) that does the
    % transformation
  % }
  The fallback case (line~\ref{line:ret_fallback}) is trivially sound.
  Let us assume towards contradiction that there is a bad trace
  $\trace' = \ProgState_0'l_0'\ProgState_1'l_1'\ldots\ProgState_n'$ of
  $\Prog' \models \ProgCons$, that is obtained by transformation of 
  $\trace = \ProgState_0l_0\ProgState_1l_1\ldots\ProgState_n$.
  For each $0 \leq i < n$, let $p(i)$ be such that the instruction at
  position $i$ in $\trace$ is at position $p(i)$ in $\trace'$ after the
  sequence of switching transformations taking $\trace$ to $\trace'$.

  If for every data-flow into assertion edge in $x \to y$ in
  $G(\trace)$, we have that $p(x) \to p(y)$ is a corresponding data-flow
  into assertion edge in $G(\trace')$, then it can be easily shown that
  $\trace'$ is also good (each corresponding edge in $\trace'$ reads the
  same values as in $\trace$).
  Now, suppose $x \to y$ is the first (with minimal $x$) such edge in
  $\trace$ that does not hold in $\trace'$.
  We will show in two steps that $p(x)$ happens before $p(y)$ in $\trace'$,
  and that $p(y)$ reads from $p(x)$ which will lead to a contradiction.

  For the first step, we know that there exists a cover of $x \to y$
  in $\trace$.
  For now, assume there is exactly one cover -- the other case is
  similar.
  For each edge $a \to b$ in this cover, no switching transformation
  can switch the order of $l_a$ and $l_b$:
  \begin{compactitem}
  \item~If $a \to b$ is a data-flow into conditionals edge,
    as $\trace'$ has to preserve all $\bluen$ edges (definition of
    regression), 
    $p(a)$ happens before $p(b)$ in $\trace'$.
  \item~If $a \to b$ is a non-free order edge, no switching
    transformation can reorder $a$ and $b$ as that would change
    variables values (by definition of non-free edges).
  \item~If $a \to b$ is a intra-thread order edge, we have that $\Prog'
    \models \ProgCons$ and $\ProgCons \implies a \leq b$, and hence, no
    switching transformation would change the order of $a$ and $b$.
  \end{compactitem}
  Hence, we have that all the happens before relations given by the
  cover are all preserved by $\trace'$ and hence, $p(a)$ happens before
  $p(a)$ in $\trace'$.
  The fact that $p(y)$ reads from $p(x)$ follows from a similar
  argument with the $\interfere(x \to y)$ edges
  showing that every interfering write either happens before $p(x)$ or
  after $p(y)$. 
  % Hence, we have that the data-flow into assertion edge $p(x) \to p(y)$
  % holds in $\trace'$ leading to a contradiction.
  \qed
\end{proof}

\subsection{Eliminating Bad Traces}
\label{sec:fix_bad}

% \arsays{Write about previous CAV paper, plus fixing the
% one-sided-communication traces}  

% The goal of this section is to find fixes for bad traces.
Given a bad trace $\trace$ of $\Prog$, a program constraint $\ProgCons$
is a {\em error eliminating constraint} if for all transformations
$\reorder_1,\ldots,\reorder_k$ and $\Prog'$ such that 
$\Prog\xrightarrow{\reorder_1\ldots\reorder_k}\Prog'$ and $\Prog'
\models \ProgCons$, each bad trace
$\trace'$ in $\reorder_k^f\circ\ldots\circ\reorder_1^f(\trace)$ is not a
trace of $\Prog'$.
In~\cite{cav2013}, we presented an algorithm to fix bad traces using
reordering and atomic sections.
The main idea behind the algorithm is as follows.
Given a bad trace $\trace$, we 
\begin{inparaenum}[(a)]
\item first, generalize the trace into a partial order trace; and
\item then, compute a program constraint that violates some essential
  part of the ordering necessary for the bug. 
\end{inparaenum}

More precisely, the procedure builds a trace elimination graph which
contain edges corresponding to the orderings necessary for the
bug to occur, as well as the edges corresponding program constraints.
Fixes are found by finding cycles in this graph -- the conjunction of
the program constraints in a cycle form an error elimination constraint.
Intuitively, the program constraints in the cycle will enforce a
happens-before conflicting with the orderings necessary for the bug.
% We illustrate this algorithm with an example.

\begin{figure}[tb]
\ffigbox
{%
\scriptsize%
  \begin{subfloatrow}[3]
    \ffigbox[\FBwidth]% Width of subfloat
    {%
      \begin{tikzpicture}
        \node[rounded rectangle] (c) {{\tt A: x:=1}};
        \node[rounded rectangle, below of=c, yshift=5] (e) {{\tt B: z:=1}};
        \node[rounded rectangle, below of=e, yshift=5] (d) {{\tt C: y:=1}};
        \node[rounded rectangle, right of=c, xshift=30, yshift=-10] (a) {{\tt 1: await(x=1)}};
        \node[rounded rectangle, below of=a, yshift=5] (b) {{\tt 2: assert(y=1)}};

        \draw[->] (c) edge (a);
        \draw[->] (b) edge (d);
	\draw[->, dotted] (a) edge node[right] {\color{red} $1 \leq 2$} (b);
	\draw[->, dotted] (d) edge[bend left,in=90,out=90] node[left] {\color{red} $C \leq A$} (c);
      \end{tikzpicture}
    }{\vspace{-3ex}\subcaption*{}\label{fig:bad_trace_ex1}}
    \ffigbox[\FBwidth-0.5cm]% Width of subfloat
    {%
      \begin{tikzpicture}
        \node[rounded rectangle] (c) {{\tt A: x:=0}};
        \node[rounded rectangle, below of=c, yshift=-20] (e) {{\tt B: x:=1}};
        \node[rounded rectangle, right of=c, xshift=30, yshift=-15] (a) {{\tt 1: assert(x=1)}};

        \draw[->] (c) to (a);
	\draw[->, draw, dotted] (e) edge[bend left] node[right] {\color{red} $[A, B]$}  (c);
        \draw[->, draw] (a) edge (e);
      \end{tikzpicture}
    } {\vspace{-3ex}\subcaption*{}\label{fig:bad_trace_ex2}}
    \ffigbox[\FBwidth-0.25cm]% Width of subfloat
    {%
      \begin{tikzpicture}
        \node[rounded rectangle] (c) {{\tt A: x:=1}};
        \node[rounded rectangle, below of=c, yshift=0] (e) {{\tt B: y:=1}};
        \node[rounded rectangle, right of=c, xshift=30, yshift=15] (a) {{\tt 1: assert(y=1)}};

        % \draw[->] (c) to (a);
        \draw[->, draw, dotted] (e) edge[bend right] node[right,xshift=2] {\color{red} $B \preceq 1$}  (a);
        \draw[->, draw] (a) edge (e);
        \draw[->, draw] (c) edge (e);
      \end{tikzpicture}
    } {\vspace{-3ex}\subcaption*{}\label{fig:bad_trace_ex3}}
    \end{subfloatrow}
  }
  {\vspace{-7ex}\caption{Eliminating bad traces}\label{fig:bad_trace}}
\end{figure}
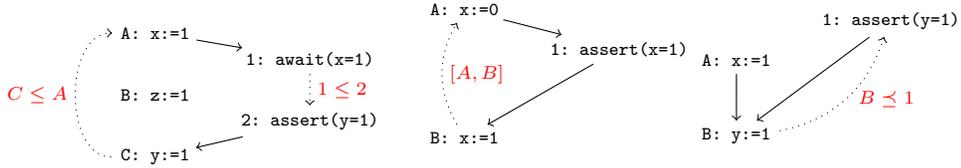
\begin{example}
  % Consider the program ${\tt \{ A: x := 1; B: z := 1; C: y := 1 \} || \{ 1:
  % await(x=1); 2: assert(y=1) \}}$.
  Consider the program in Figure~\ref{fig:bad_trace}(left) and the trace
  elimination graph for the trace $A; B; 1; 2; C$.
  The orderings $A$ happens-before $1$ and $2$
  happens-before $C$ are necessary for the error to happen.
  The cycle $C \to A \to 1 \to 2 \to C$ is the elimination cycle.
  The corresponding error eliminating constraint is $C \leq
  A \wedge 1 \leq 2$, and one possible fix is to move $C$ ahead of $A$.
  For the bad trace $A; 1; B$ in Figure~\ref{fig:bad_trace}(center), the
  elimination cycle is $A \to 1 \to B\to A$ giving us the 
  constraint $[ A; B]$ and an atomic section around $A; B$ as the
  fix.
\end{example}

\paragraph{The $\FixBad$ algorithm.}
The $\FixBad$ algorithm takes as input a program $\Prog$, a constraint
$\ProgCons$ and a bad trace $\trace$.
It outputs a program constraint $\ProgCons'$,
sequence of program transformations $\reorder_1,\ldots,\reorder_k$, and
a new program $\Prog'$, such that
$\Prog\xrightarrow{\reorder_1\ldots\reorder_k}\Prog'$.
The algorithm guarantees that
\begin{inparaenum}[(a)]
 \item~$\ProgCons'$ is an error eliminating constraint;
 \item $\Prog' \models \ProgCons \wedge \Prog' \models \ProgCons'$; and
 \item if there is no preemption-free trace $\trace'$ of $\Prog$ such
   that $\trace$ freely transforms to $\trace'$ (i.e., $\trace' \in
   f(\trace)$), then none of the transformations $\reorder \in
   \{\reorder_1,\ldots,\reorder_k\}$ acts across preemption-points.
\end{inparaenum}
The fact that $\reorder_1\ldots\reorder_k$ and $\Prog'$ can be chosen to
satisfy (c)~is a consequence of the algorithm described
in~\cite{cav2013}.

\noindent{\bf Fixes using wait/notify statements.}
Some programs cannot be fixed by statement reordering or atomic section
insertion.
These programs are in general outside our definition of the program 
repair problem as they have bad sequential traces.
However, they can be fixed by the insertion of wait/notify
statements.
One such example is depicted in Figure~\ref{fig:bad_trace}(right) where the
trace $1; A; B$ causes an assertion failure.
A possible fix is to add a {\tt wait} statement before $1$ and a
corresponding {\tt notify} statement after $B$.
The algorithm $\FixBad$ can be modified to insert such wait-notify
statements by also considering constraints of the
form $X \preceq Y$ to represent that $X$ is scheduled before $Y$ -- the
corresponding program transfomation is to add a wait statement before
$Y$ and a notify statement after $X$.
In Figure~\ref{fig:bad_trace}(right), the edge $B \to 1$ represents such a
constraint $B \preceq 1$ -- the elimination cycle $1 \to B
\to 1$ corresponds to the above described fix.

\section{The Program-Repair Algorithm}
\label{sec:algo}

\begin{algorithm}[tbh]
  \begin{algorithmic}[1]
    \REQUIRE A concurrent program $\Prog$, all sequential traces are good
    \ENSURE Program $\Prog^*$ such that $\Prog^*$ has no bad traces
    \STATE $\ProgCons\gets \mathit{true}; T_G \gets \emptyset$
    \WHILE{\TRUE}
    \STATE{$\mathbf{if}~\mathit{Verify}(\Prog) = \TRUE~\mathbf{then~return}~\Prog$} \label{line:verify_prog}
    \STATE Choose $\trace$ from $\Prog$ \hfill (non-deterministic) \label{line:pick_trace}
    \IF{$\trace$ is non-erroneous}
    \STATE $\ProgCons \gets \ProgCons \wedge \LearnGood(\trace); T_G
    \gets T_G\cup \{\trace\}$
    \ELSE
    \STATE $([\reorder_1,\ldots,\reorder_k],\Prog,\ProgCons') \gets
    \FixBad(\Prog, \PProg, \ProgCons, \trace); \quad \ProgCons \gets \ProgCons \wedge \ProgCons'$\label{line:fix_bad}
    \STATE
    $T_G\gets \bigcup_{\trace_g\in T_G} \{\trace'_g|\trace'_g\in\reorder_k\circ\ldots\circ\reorder_1(\trace^g)\wedge \trace'_g\in\Prog\}$ \label{line:adjust}
    \ENDIF
    \ENDWHILE
  \end{algorithmic}
  \caption{Program-Repair Algorithm for Concurrency}
  \label{algo:main}
\end{algorithm}

Algorithm~\ref{algo:main} is a program-repair procedure to fix
concurrency bugs while avoiding regressions.
The algorithm maintains the current program $\Prog$,
and a constraint $\ProgCons$ that restricts possible reorderings.
In each iteration, the algorithm tests if $\Prog$ is correct and if so
returns $\Prog$.
If not it picks a trace $\trace$ in $\Prog$
(line~\ref{line:pick_trace}).
If the trace is good it learns the regression-preventing constraint
$\ProgCons$ for $\trace$ and the trace $\trace$ is added to the set of
good traces $T_G$ ($T_G$ is required only for the correctness proof).
If $\trace$ is bad it calls $\FixBad$ to generate a new program that
excludes $\trace$ while respecting $\ProgCons$, and $\ProgCons$ is
strengthened by conjunction with the error elimination constraint
$\ProgCons'$ produced by $\FixBad$.
% The set $T_G$ is adjusted to the changed $\Prog$ (line~\ref{line:adjust}).
The algorithm terminates with a valid solution for all choices
of $\Prog'$ in line~\ref{line:fix_bad} as the constraint $\ProgCons$ is
strengthened in each $\FixBad$ iteration. 
Eventually, the strongest program-constraint will restrict the possible
program $\Prog'$ to one with large enough atomic sections such that it
will have only preemption-free or sequential traces.
% The following theorems state the correctness properties of the
% algorithm.

\begin{theorem}[Soundness]
  \label{lem:algo_soundness}
  Given a program $\Prog$, Algorithm~\ref{algo:main} returns a program
  $\Prog'$ with no bad traces that preserves the sequential and correct
  preemption-free behaviour of $\Prog$.
  Further, each iteration of the {\bf while} loop where a bad trace
  $\trace$ is chosen performs a regression-free error fix with respect
  to the good traces $T_G$.
\end{theorem}
\remove{
\begin{proof}[outline]
  We need to show there are no bad traces in $\Prog'$ and all sequential traces remain.
  The first follows because the algorithm only terminates if $\Prog'$ is correct and
  the second follows because we do not insert wait-notify statements. Additionally we need to show that 
  if all preemption-free traces are good they are preserved. This is guaranteed by $\FixBad$
  as it does not do transformations across preemption-points if the trace is not preemption-free.
  
  We further need to show that no regression is introduced in each iteration and that the
  bad trace disappears. These follow directly from the properties of $\FixBad$.
\end{proof}
}
The extension of the $\FixBad$ algorithm to wait/notify
fixes in Algorithm~\ref{algo:main} may lead to $\Prog'$ not preserving
the good preemption-free and sequential behaviours of $\Prog$.
However, in this case, the input $\Prog$ violates the pre-conditions 
of the algorithm.

\begin{theorem}[Fair Termination]
  Assuming that a bad trace will eventually be chosen in 
  line~\ref{line:pick_trace} if one exists in $\Prog$, Algorithm~\ref{algo:main}
  terminates for any instantiation of $\FixBad$.
\end{theorem}

\newcommand\choices{\mathcal{C}}
\newcommand\choice{\mathbf{c}}
\newcommand\expts{\mathcal{E}}
\newcommand\inp{\mathbf{i}}
\mypara{A Generic Program-Repair Algorithm.}
We now explain how our program-repair algorithm relates to generic
synthesis procedures based on {\em counter-example guided inductive
synthesis} (CEGIS)~\cite{asplos06}.
In the CEGIS approach, the input is a {\em partial-program} $\PProg$, i.e., a
non-deterministic program and the goal is to specialize $\PProg$ to a
program $\Prog$ so that all behaviours of $\Prog$ satisfy a
specification.
In our case, the partial-program would non-deterministically choose
between various reorderings and atomics sections.
Let $\choices$ be the set of choices (e.g., statement
orderings) available in $\PProg$.
For a given $\choice \in \choices$, let $\mathbb{P}(\PProg, \choice,
\inp)$ be the predicate that program obtained by
specializing $\PProg$ with $\choice$ behaves correctly on the input
$\inp$.

The CEGIS algorithm maintains a set $\expts$ of inputs 
called experiments.
In each iteration, it finds $\choice^* \in \choices$ such
that the $\forall \inp \in \expts: \mathbb{P}(\PProg, \choice^*, \inp)$.
Then, it attempts to find an input $\inp^*$ such that
$\mathbb{P}(\choice^*, \inp^*)$ does not hold.
If there is no such input, then $\choice^*$ is the correct
specialization.
Otherwise, $\inp^*$ is added to $\expts$.
This procedure is illustrated in Figure~\ref{fig:cegis}(left).
Alternatively, CEGIS can be rewritten in terms of
constraints on $\choices$.
For each input $\inp$, we associate the constraint $\phi_\inp$
where $\phi_\inp(\choice) \Leftrightarrow \mathbb{P}(\PProg, \choice,
\inp)$.
Now, instead of $\expts$, the algorithm maintains the constraint $\Phi =
\bigwedge_{\inp \in \expts} \phi_\inp$.
Every iteration, the algorithm picks a $\choice$ such that $\choice
\models \Phi$; tries to find an input $\inp^*$ such that
$\neg \mathbb{P}(\PProg, \choice, \inp)$ holds, and then strengthens
$\Phi$ by $\phi_{\inp^*}$.

\begin{figure}
  \begin{minipage}{0.3\textwidth}
    \begin{tikzpicture}
      \node [draw, rectangle] (a) {
	\tabular{c} $\exists? \choice^* : \bigwedge_{\inp \in E} $ \\
	$ \mathbb{P}(\PProg, \choice^*, \inp) $
	\endtabular
      };
      \node [draw, below of=a, yshift=-5, rectangle] (b) {
	\tabular{c}
	$\exists? \inp^*$ s.t.   \\
	$\neg \mathbb{P}(\PProg, \choice^*, \inp^*)$
	\endtabular
      };

      \draw[->] (a) edge[bend left, out=90, in=90] (b);
      \draw[->] (b) edge[bend left, out=90, in=90] node[left] {
	\tabular{c} 
	$E = E \cup $\\
	$\{ \inp^* \}$
	\endtabular
      } (a);
    \end{tikzpicture}
  \end{minipage}
  \hfill
  \begin{minipage}{0.6\textwidth}
    \begin{tikzpicture}
      \node [draw, rectangle] (a) {
	\tabular{c} $\exists? \choice^* : \choice^* \models \Phi$ 
	% $ \mathbb{P}(\PProg, \choice^*, \inp) $
	\endtabular
      };
      \node [draw, below of=a, yshift=-5, xshift=-37, rectangle] (b) {
	\tabular{c}
	$\exists? \inp^*$ s.t.   \\
	$\neg \mathbb{P}(\PProg, \choice^*, \inp^*)$
	\endtabular
      };
      \node [draw, below of=a, yshift=-5, xshift=37, rectangle] (c) {
	\tabular{c}
	$\exists? \inp^*$ s.t.   \\
	$\mathbb{P}(\PProg, \choice^*, \inp^*)$
	\endtabular
      };

      \draw[->] (a) edge (b);
      \draw[->] (a) edge (c);
      \draw[->] (b) edge[out=135,in=270,bend left] node[left] {
	\tabular{c}
	$\Phi = \Phi \wedge$ \\
	$\FixBad(\inp^*)$ 
	\endtabular
      } (a.west);
      \draw[->] (c) edge[out=135,in=180,bend right] node[right] {
	\tabular{c}
	$\Phi = \Phi \wedge $ \\
	$\LearnGood(\inp^*)$ 
	\endtabular
      } (a.east);
    \end{tikzpicture}
  \end{minipage}
  \vspace{-2ex}
  \caption{The CEGIS and PACES spectrum}
  \label{fig:cegis}
\end{figure}
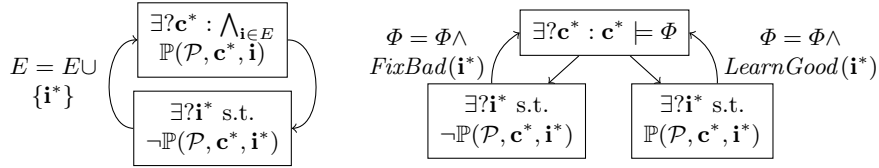

This procedure is exactly the else branch (i.e., $\FixBad$ procedure) of
an iteration in Algorithm~\ref{algo:main} where $\inp^*$
and $\phi_{\inp^*}$ correspond to $\trace$ and $\FixBad(\trace)$.
Intuitively, the initial variable values in $\trace$ and the scheduler
choices are the inputs to our concurrent programs.
This suggests that the then branch in Algorithm~\ref{algo:main} could
also be incorporated into the standard CEGIS approach.
This extension (dubbed PACES for {\em Positive and Counter-Examples 
in Synthesis}) to the CEGIS approach is shown in
Figure~\ref{fig:cegis}(right).
Here, the algorithm in each iteration may choose to find an input for
which the program is correct and use the constraints arising
from it.
We discuss the advantages and disadvantages of this approach below.
% that arise from the
% use of both examples and counter-example inputs, as well as the use of
% constraints instead of explicit inputs.

\vspace{-1ex}
\paragraph{Constraints vs. Inputs.}
A major advantage of using constraints instead of sample inputs is the
possibility of using over- and under-approximations.
As seen in Section~\ref{sec:learn_good}, it is sometimes easier
to work with approximations of constraints due to simplicity
of representation at the cost of potentially missing good solutions.
% However, the disadvantage is that under-approximating constraints
% introduces the possibility of missing potential solutions, while
% over-approximating the constraints may lead to either introducing
% regressions or re-examining some inputs.
% \arsays{Constraints can be approximated -- inputs can't -- probably the
% main point of constraints}
%
Another advantage is that the sample inputs may have no simple
representations in some domains. 
The scheduler decisions are one such example -- the
scheduler choices for one program are hard to translate into the
scheduler choices for another.
For example, the original CEGIS for concurrency work~\cite{SLJB08} uses 
ad-hoc trace projection to translate the scheduler choices
between programs.
% On the other hand, there are many domains where the constraints are 
% harder to express or are more complex than explicit sample inputs -- for
% example, with finite combinatorial programs.
% Hence, the choice of using constraints vs. sample inputs depends heavily
% on the domain.
% \arsays{one example here}
% \arsays{However, constraints are harder to maintain. Not all domains
% have easily expressible constraints}

\vspace{-1ex}
\paragraph{Positive-examples and Counter-examples vs. Counter-examples.}
% The use of both examples and counter-examples in PACES is advantageous
% in many cases.
In standard program-repair tasks, although the faulty
program and the search space $\choices$ may be large, the
solution program is usually ``near'' the original program, i.e., the fix
is small.
Further, we do not want to change the given program unnecessarily. 
In this case, the use of positive examples and over-approximations of
learned constraints can be used to narrow down the search space quickly.
Another possible advantage comes in the case where the search space
for synthesis is structured (for example, in modular synthesis).
In this case, we can use the correct behaviour displayed by a candidate
solution to fix parts of the search space.
% \arsays{complete this
% -- Search space is structured / Modular synthesis (search space is split
% into multiple almost-independent parts). From the good behaviour in one
% part of the search space, fix parts of the structure
% }

\section{Implementation and Experiments}
\label{sec:impl}

We implemented Algorithm~\ref{algo:main} in our tool ConRepair.
The tool consists of 3300 lines of Scala code and is available at
https://github.com/thorstent/ConRepair.
Model checker CBMC~\cite{cbmc} is used for generating both good and bad
traces, and on an average more than 95\% of the total execution time is
spent in CBMC.
Model checking is far from optimal to obtain good traces, and we expect
that techniques from~\cite{Sen:2008:RDR:1375581.1375584} can be used to
generate good traces much faster.
Our tool can operate in two modes: In ``mixed'' mode it first
analyses good traces and then proceeds to fixing the program. 
The baseline ``badOnly'' mode skips the analysis of good traces
(corresponds to the algorithm in~\cite{cav2013}).

In practice the analysis of bad traces usually generates a large number
of potential reorderings that could fix the bug.
Our original algorithm from~\cite{cav2013} (badOnly {\tt ce1}) prefers
reorderings over atomic sections, but in examples where an atomic
section is the only fix, this algorithm has poor performance.
To address this we implemented a heuristic ({\tt ce2}) that places
atomic sections before having tried all possible reorderings, but this
can result in solutions having unnecessary atomic sections.

The fall back case in Algorithm~\ref{algo:learn_good} severely limits
further fixes -- it forces further fixes involving the same instructions
to be atomic sections.
Hence, in our implementation, we omit this step and prefer an unsound
algorithm (i.e., not necessarily regression-free) that can fix more
programs with reorderings.
While the implemented algorithm is unsound, our experiments show that
even without the fallback, in our examples, there is no regression
except for one artificial example ({\tt ex-regr.c}) constructed
precisely for that purpose.

\vspace{-1ex}
\paragraph{Benchmarks.}
We evaluate our tool on a set of examples that model real bugs found and 
fixed in Linux device drivers by their developers.  
To this end, we explored a history of bug fixes in the drivers subtree
of the Linux kernel and identified concurrency bugs.
% We used bug descriptions as well as code inspection to identify 
% concurrency bugs.
We further focused our attention on a subset of 
particularly subtle bugs involving more than two racing threads 
and/or a mix of different synchronization mechanisms, e.g., lock-based 
and lock-free synchronization.  
Approximately 20\% of concurrency bugs that we considered satisfy this
criterion.
Such bugs are particularly tricky to fix either manually or
automatically, as new races or deadlocks can be easily introduced while
eliminating them.
% (deadlocks are modelled using assertions because CBMC cannot detect
% deadlocks by itself). 
Hence, these bugs are most likely to benefit from good trace analysis.

\begin{wraptable}{r}{0.65\textwidth}
  \small
  \begin{tabular}{l|l|l|l|l}
     % & &  & \multicolumn{2}{c}{badOnly}\\
    File & LOC & mixed & {\scriptsize badOnly ce1} & {\scriptsize badOnly ce2} \\
    \hline
    {\tt ex1.c} & 60 & 1 & 2 & 2  \\
    {\tt ex2.c} & 37 & 2 & 5 & 6 \\
    {\tt ex3.c} & 35 & 1 & 2 & 2 \\
    {\tt ex4.c} & 60 & 1 & 2 & 2 \\
    {\tt ex5.c} & 43 & 1 & 8 & 3 \\
    {\tt ex-regr.c} & 30 & 2 & 2 & 2\\
    {\tt paper1.c} & 28 & 1 & 3 & 3\footnote{{\tt ce2} heuristic placed 
  unnecessary atomic section\vspace{-3ex}} \\
    {\tt dv1394.c} & 81 & 1 (13+4s) & 51 (60s) & 5\mpfootnotemark[1] (9s) \\
    {\tt iwl3945.c} & 66 & 1(3+2s) & 2(2s) & 2(2s) \\
    {\tt lc-rc.c} & 40 & 10 (2+7s) & 179 (122s) & 203 (134s) \\
    % \hline
    {\tt rtl8169.c} & 405 & 7 (10+45m) & \textgreater100 (\textgreater6h) & 8 (54m) \\
    {\tt usb-serial.c} & 410 & 4 (56+20m) & 6 (38m) & 6 (38m)
  \end{tabular}
  \vspace{-2ex}
  \caption{Results in iterations and time needed.}
  \label{tbl:results}
\end{wraptable}
Table~\ref{tbl:results} shows our experimental results: the iterations
and the wall-clock time needed to find a valid fix for our mixed
algorithm and the two heuristics of the badOnly algorithm. For the mixed
algorithm the time is split into the time needed to generate and
analyse good traces (first number) and the time needed for the fixing
afterwards.

\vspace{-1ex}
\paragraph{Detailed analysis.}
The artificial examples {\tt ex1.c} to {\tt ex5.c} are used for testing
and take only a few seconds;
example {\tt paper1.c} is the one in Figure~\ref{fig:example}.
%
% {\new
Example {\tt ex-regr.c} was constructed to show unsoundness of the
implementation.
% In this example, 
% the assertion contains a disjunction and
% The data-flow into assertions analysis generates a number of
% edges that cannot be covered.
% over-approximation of real dependencies
% that cannot be covered.
% date-flow edges into assertions that cannot be covered.
% As our implementation does not fallback this example allows for a
% regression.
%
% }
%
Example {\tt usb-serial.c} models the USB-to-serial adapter driver.
Here, from the good traces the tool learns that two statements should
not be reordered as it will trigger another bug.
This prompts them to be reordered above a third statement together,
while the badOnly analysis would first move one, find a new bug, and
then fix that by moving the other statement.
Thus, the good trace analysis saves us two rounds of bug fixing and
reduces bug fixing time by 18 minutes.
% However, this does not offset the high initial cost due to the 
% good traces being generated by a model checker.
% This is because of the
% high number of assertions in the model with just one actual bug and
% due to CBMC taking long to find good traces.

%The examples {\tt dv1394.c} and {\tt rtl8169.c} emphasize the
%trade-off involved in  
%choosing between the new algorithm and  the {\tt ce2} heuristic in
%``badOnly'' mode. 
The  {\tt rtl8169.c} example
models the Realtek 8169 driver containing 5 concurrency bugs.  One of
the reorderings that the tool considers introduces a new bug; further, after
doing the reordering, the atomic section is the only valid fix. 
The good trace analysis discover that the reordering would lead to a
new bug, and thus does the algorithm does not use it.
But, without good traces, the tool uses the faultly reordering and then
{\tt ce1} takes a very long time to search through all possible reorderings
and then discover that an atomic section is required.
The situation is improved when using heuristic {\tt ce2} as it
interrupts the search early.  However, the same heuristic has an adverse 
effect in the {\tt dv1394.c} example: by interrupting  
the search early, it prevents the algorithm from finding a correct 
reordering and inserts an unnecessary atomic section.
The {\tt dv1394.c} example also benefits from good traces in a different
way than the other examples.
Instead of preventing regressions, they are used
to obtain {\em hints} as to what reorderings would provide coverage
for a specific data-flow into assertion edge. 
Then, if a bad trace is encountered and can be fixed by the hinted
reordering, the hinted reordering is preferred over all other possible
ones.
% The hints are then used if a bad trace is encountered to prefer the
% hinted reorderings if they are among the possible solutions.
Without hints the {\tt dv1394.c} example
would require 5 iterations. Though hints are not part of our theory they
are a simple and logical extension.

%In example {\tt lc-rc.c}, the heuristic {\tt ce2} does not help as
%the bug cannot be fixed using atomic sections.
Example {\tt lc-rc.c} models a bug in an ultra-wide 
band driver that requires two reorderings to fix. Though there is initially 
no deadlock, one may easily be introduced when reordering statements.
Here, the good-trace analysis identifies a dependency 
between two {\tt await} statements and learns not to reorder statements
to prevent a deadlock. 
Without good traces, a large number of candidate solutions that cause
a regression are generated.  

%\noindent{\bf Conclusion.} 
\vspace{-1ex}
\section{Conclusion}
We have developed a regression-free algorithm for fixing errors that
are due to concurrent execution of the program. The contributions
include the problem setup (the definitions of program repair for
concurrency, and the regression-free algorithm), the PACES approach
that extends the CEGIS loop with learning from positive examples, and
the analysis of positive examples using data flow to assertions and to
synchronization constructs.   

There are several possible directions for future work. 
One interesting direction is to examine the possibility of
extending the definition of regressions (see
Remark~\ref{rem:regressions} and Example~\ref{ex:non_locality}) -- this
requires going beyond data-flow analysis for learning
regression-preventing constraints.
Another possible extension is to remove the assumption that the
errors are data-independent.
%A third direction is to develop a sound and complete
%algorithm for learning regression-preventing 
%constraints. 
A more pragmatic goal would be to develop a practical version of the
tool for device-driver synthesis starting from the current prototype.
% These include investigating learning sound and complete
% regression preventing constraints, improving the trace generalization
% algorithm, and developing a practical version of the tool starting 
% from the current prototype version. 
%Future work. 

\noindent {\bf Acknowledgements.} We would like to thank Daniel
Kroening and Michael Tautschnig for their prompt help with all our
questions about CBMC. 
We would also like to thank Roderick Bloem, Bettina K\"onighofer and Roopsha Samanta for
fruitful discussions 
regarding repair of concurrent programs.

\bibliographystyle{splncs03}
\bibliography{references}

\newpage

\appendix

\section{The iwl3945 driver.}\label{app:driver}

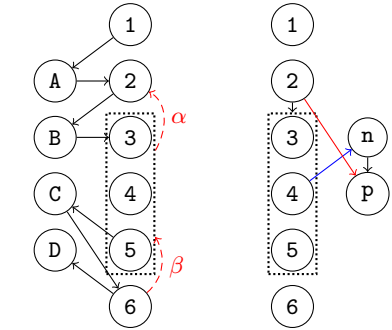
\begin{figure}[tb]
\ffigbox
{%
  \begin{subfloatrow}[1]
    \ffigbox[\FBwidth+0.5cm]% Width of subfloat
    {%
      \begin{minipage}{0.35\textwidth}
      \small{ \tt
config\_thread () \{\\
\phantom{~~}A: lock(rtnl\_lock);\\
\phantom{~~}B: lock(lock);\\
\phantom{~~}C: unlock(lock);\\
\phantom{~~}D: unlock(rtnl\_lock);\\
\}
      }
      \end{minipage}
\quad
      \begin{minipage}{0.4\textwidth}
      \small{ \tt
iwl3945\_bg\_alive\_start\_thread() \{\\
\phantom{~~}1: lock(lock);\\
\phantom{~~}2: restart = 1;\\
\phantom{~~}$\begin{bmatrix*}[l]
  \mbox{3: lock(rtnl\_lock);}\\
  \mbox{4: notify = 1;}\\
  \mbox{5: unlock(rtnl\_lock);}
 \end{bmatrix*}$\\
\phantom{~~}6: unlock(lock);\\
\}
      }
      \end{minipage}
\quad
      \begin{minipage}{0.4\textwidth}
      \small{ \tt
reassoc\_thread () \{\\
\phantom{~~}n: await(notify==1);\\
\phantom{~~}p: assert(restart==1);\\
\}
      }
      \end{minipage}
    }
    {%
      \subcaption{Simplified threads of the iwl3945 driver}\label{fig:iwl3945}%
    }
  \end{subfloatrow}\\
  \begin{subfloatrow}[2]

    \ffigbox[\FBwidth+0.5cm]% Width of subfloat
    {%
    \begin{tikzpicture}[%
    hhilit/.style={draw=black, thick, densely dotted,
    inner xsep=.1em,
    inner ysep=.1em},
    ]
      \tikzstyle{every state}=[draw=black, minimum size=12]

      \node[state]  (1) {{\tt 1}};
      \node[state,below of=1,yshift=0.25cm]  (2) {{\tt 2}};
      \node[state,below of=2,yshift=0.25cm]  (3) {{\tt 3}};
      \node[state,below of=3,yshift=0.25cm]  (4) {{\tt 4}};
      \node[state,below of=4,yshift=0.25cm]  (5) {{\tt 5}};
      \node[state,below of=5,yshift=0.25cm]  (6) {{\tt 6}};
      
      \node[state,left of=2]  (a) {{\tt A}};
      \node[state,below of=a,yshift=0.25cm]  (b) {{\tt B}};
      \node[state,below of=b,yshift=0.25cm]  (c) {{\tt C}};
      \node[state,below of=c,yshift=0.25cm]  (d) {{\tt D}};

      \node[hhilit, fit=(3) (5)] (bo) {};
      
      \path[->]
            (1) edge (a)
            (a) edge (2)
            (2) edge (b)
            (b) edge (3)
            (5) edge (c)
            (c) edge (6)
            (6) edge (d)
            (bo) edge[red,densely dashed,in=330,out=60] node[xshift=0.2cm] {$\alpha$} (2)
            (6) edge[red,densely dashed,in=300,out=40] node[xshift=0.2cm] {$\beta$} (bo)
      ;
      
    \end{tikzpicture}
    }
    {%
      \subcaption{Possible fixes for this trace}\label{fig:bad}%
    }
    
    \ffigbox[\FBwidth+0.5cm]% Width of subfloat
    {%
    \begin{tikzpicture}[%
    hhilit/.style={draw=black, thick, densely dotted,
    inner xsep=.1em,
    inner ysep=.1em},
    ]
      \tikzstyle{every state}=[draw=black, minimum size=12]
      \node[state]  (1) {{\tt 1}};
      \node[state,below of=1,yshift=0.25cm]  (2) {{\tt 2}};
      \node[state,below of=2,yshift=0.25cm]  (3) {{\tt 3}};
      \node[state,below of=3,yshift=0.25cm]  (4) {{\tt 4}};
      \node[state,below of=4,yshift=0.25cm]  (5) {{\tt 5}};
      \node[state,below of=5,yshift=0.25cm]  (6) {{\tt 6}};

      \node[state,right of=3]  (n) {{\tt n}};
      \node[state,below of=n,yshift=0.25cm]  (p) {{\tt p}};
      
      \node[hhilit, fit=(3) (5)] (bo) {};
      
      \path[->]
            (4) edge[blue]  (n)
            (2) edge[red] (p)
            (n) edge (p)
            (2) edge (bo)
      ;

    \end{tikzpicture}
    }
    {%
      \subcaption{Learning from a good trace}\label{fig:good}%
    }
  \end{subfloatrow}
}
{\caption{Model of deadlock bug in the iwl3945 driver}}
\end{figure}

This example models a real concurrency bug found in the 
Linux driver for the Intel wireless adapter 3945.  The model involves three 
kernel threads calling driver entry points shown in Figure~\ref{fig:iwl3945}.

The driver suffers from a classical ABBA deadlock; if the {\tt config\_thread}
locks {\tt rtnl\_lock} then the {\tt iwl3945\_bg\_alive\_start\_thread} locks
{\tt lock} none of the threads can proceed.

{\tt config\_thread} is located outside the driver and cannot be changed
in order to fix the bug. Furthermore instructions {\tt 3} to {\tt 5} need
to stay together because they are located in a separate function in the
original code.
The {\tt iwl3945\_bg\_alive\_start\_thread} is used to restart the device
if it no longer responds properly. After a restart (line {\tt 2}) the {\tt
reassoc\_thread} is notified (line {\tt 4}), which depends on the device
having completed the restart.

In the actual model the deadlock is modelled with an assertion
between lines {\tt A} and {\tt B} that
tests which thread owns which locks and which thread is waiting for
which locks. The assertion fails if there is a deadlock. This construction
is needed because CBMC does not support deadlock detection.

Our old algorithm {\tt ce1} without a learning from good traces phase now
proceeds to finding a bad trace, such as 
$1\to A\to 2\to B\to 3\to 4\to 5\to C\to 6\to D$ displayed
in Fig.~\ref{fig:bad}. Note that after an assertion failure the
trace is completed as if the assertion had succeeded in order to find all
reordering options. This leaves a number of bug-fixes to the algorithm: 
We denote only the two important for the description of our algorithms.
The first is to  move
the block {\tt3-5} in front of {\tt 1} (denoted as $\alpha$)
and the second is to move {\tt6} in front of {\tt3-5} (denoted as $\beta$).
Without further knowledge the algorithm will choose $\alpha$ for no
particular reason.
This does not fix the deadlock, but it furthermore introduces a regression because
now assertion {\tt p} may fail because the notify signal is sent before the restart
is completed.

The good trace analysis prevents a regression with respect to good traces
such as the trace where all threads are run sequentially. The relevant parts of the
result of our analysis are depicted in Fig.~\ref{fig:good}.
The red and blue edge indicate the variable assignments the assertion
and the await read from respectively. In order to protect the red edge
two black edges are needed, from {\tt2} to the instruction block
{\tt3-5} and from {\tt n} to {\tt p}. When the algorithm comes to the phase of fixing 
bugs it will discover the same two possibilities, but $\alpha$ is blocked
by the black edge from {\tt2} to {\tt3-5}. The algorithm will then
choose the correct fix $\beta$.

This is also the fix the developers took in the actual driver.
By taking the notify with its needed {\tt rtnl\_lock} out of the {\tt lock} environment
the deadlock is avoided.

\end{document}